\documentclass[conference, 12pt, twocolumn]{ieeeconf}

\IEEEoverridecommandlockouts
\usepackage[usenames]{color}
\usepackage{enumerate}
\usepackage{url}
\usepackage{subfigure}
\usepackage{amsfonts,mathrsfs}
\usepackage{amssymb,amsmath}
\usepackage{verbatim}
\usepackage{acronym}
\usepackage{mathtools}
\usepackage{cite}
\usepackage{graphicx}
\usepackage{algorithm}
\usepackage[noend]{algpseudocode}

\def\fskip#1{}

\newtheorem{theorem}{Theorem}

\newtheorem{definition}{Definition}

\newtheorem{lemma}{Lemma}

\newtheorem{remark}{Remark}

\renewcommand{\theenumi}{\arabic{enumi}}

\def\1{{\bf 1}}

\newcommand{\remove}[1]{}

\def\argmax{\mathop{\rm argmax}}

\begin{document}
\title{Open-Loop Equilibrium Strategies for Dynamic Influence Maximization Game Over Social Networks}
\author{\authorblockN{S. Rasoul Etesami*}
 \authorblockA{\vspace{-1.2cm}
}
\thanks{*Department of Industrial and Systems Engineering,  Coordinated Science Lab,  University of Illinois at Urbana-Champaign,  Urbana, IL 61801 (etesami1@illinois.edu).  This work is supported by the NSF CAREER Award under Grant No. EPCN-1944403.}
}
\maketitle
\begin{abstract}
We consider the problem of budget allocation for competitive influence maximization over social networks.  In this problem, multiple competing parties (players) want to distribute their limited advertising resources over a set of social individuals to maximize their long-run cumulative payoffs.  It is assumed that the individuals are connected via a social network and update their opinions based on the classical DeGroot model.  The players must decide the budget distribution among the individuals at a finite number of campaign times to maximize their overall payoff given as a function of individuals' opinions.  We show that i) the optimal investment strategy for the case of a single-player can be found in polynomial time by solving a concave program, and ii) the open-loop equilibrium strategies for the multiplayer dynamic game can be computed efficiently by following natural regret minimization dynamics.  Our results extend the earlier work on the static version of the problem to a dynamic multistage game.
\end{abstract}

\section{Introduction}
Due to the rapid proliferation of online social networks such as Facebook and Twitter,  the role of advertising strategies to influence public opinion has elevated to an entirely new level.  Such external influences mainly posed by online platforms, social media, political leaders, or product companies have emerged in many socioeconomic networks. For instance, political leaders often advertise their opinions through repeated campaigns among social individuals to win the election.  A common fact about all such advertising problems is that the evolution of public opinion comprises of two factors: i) the advertising or investment strategies of external influencers such as companies or political leaders (hereafter referred to as players), and ii) the internal interactions among social entities (hereafter referred to as individuals). Therefore, a major challenge for the players is to deploy investment strategies to shift public opinion toward their objectives.  As players often have limited budgets, they must choose their investment strategies while respecting certain budget constraints.

Opinion influence maximization has been extensively studied in the past literature, mainly using static optimization or static games.  For instance,  the authors of \cite{ahmadinejad2015forming} consider an optimization problem to find influential people in a social network and employ them to spread a desired behavior. A closely related problem is the seed selection problem \cite{ahmadinejad2015forming,kempe2003maximizing,gionis2013opinion,garimella2017balancing,chen2010scalable},  in which the goal is to select a limited subset of nodes in a social network as initial seeds of a specific diffusion process to influence the highest number of individuals at the end of the process. The seed selection optimization problems are typically NP-hard, although, under certain assumptions, constant approximation algorithms are known \cite{ahmadinejad2015forming,kempe2003maximizing,chen2010scalable}. 

The competitive version of the seed selection problem in which multiple influencers want to maximize their social influence has been studied using static noncooperative games \cite{alon2010note,etesami2016complexity,grabisch2018strategic,maehara2015budget,christia2021scalable,goyal2019competitive}. For instance, \cite{alon2010note} proposes a diffusion game to study the spread of different products on a social network as a result of seed selection strategies.  However, as shown in \cite{etesami2016complexity},  even deciding whether such a diffusion game admits a pure-strategy Nash equilibrium is an NP-complete problem.  Since the existence of pure-strategy Nash equilibria is a desirable property of an influence maximization game (especially for marketing strategies), different variants of influence maximization games have been proposed in the past literature \cite{maehara2015budget,christia2021scalable}.  For instance,  \cite{christia2021scalable} relaxes the discrete nature of the seed selection strategies to continuous budget allocation strategies and uses some results from socially concave games \cite{even2009convergence} to establish the existence of a pure-strategy Nash equilibrium.  Moreover,  the authors of \cite{maehara2015budget} use integral budget allocation strategies and model the influence game as a potential game to show the existence of pure Nash equilibria.  

Although the problem of competitive influence maximization has been well-understood from \emph{static} game-theoretic perspective (see,  e.g.,  \cite{grabisch2018strategic,masucci2014strategic,bindel2015bad,varma2018marketing,goyal2019competitive}),  the results based on \emph{dynamic} games are scarce and rely on simplified models.  For instance,  \cite{wang2020controlling,niazi2016consensus,niazi2020differential} use linear-quadratic differential games to model and control consensus formation among individuals by deploying advertising strategies.  Those works use maximum principle and Hamilton-Jacobi-Bellman (HJB) equations to characterize the advertising equilibrium trajectories.  Unfortunately, extending those results is difficult because solving the HJB equation or the maximum principle for nonquadratic differential games is much more complex.  

Perhaps, our work is most related to \cite{varma2019allocating}, in which the authors consider a multistage influence maximization game formulated as a two-player zero-sum game with a hybrid state process. Although that formulation provides a well-justified model for dynamic competitive influence maximization \cite{varma2019allocating,christia2021scalable}, obtaining its equilibrium strategies poses substantial challenges due to the nonquadratic structure of the cost functions and nonlinear state jump dynamics.  This work takes the first step toward solving such a dynamic hybrid game and obtains an efficient algorithm for solving its \emph{open-loop} equilibrium strategies that satisfy certain convexity assumptions.  We should mention that computing a \emph{closed-loop} equilibrium strategy in the online sequential version of our game is a more challenging task, which we leave as a future research direction.

The paper is organized as follows.  In Section \ref{sec:problem-formulation}, we formulate the problem.  In Section \ref{sec:single-agent},  we address optimal investment strategies for the case of a single player with concave stage utility functions.  In Section \ref{sec:multi-player},  we develop an efficient algorithm to obtain open-loop equilibrium strategies for multiple players with convex stage utility functions that satisfy a certain socially concave property. We conclude the paper in Section \ref{sec:conclusion}, and relegate omitted proofs to Appendix I. 

\section{Influence Maximization Game}\label{sec:problem-formulation}

In this section, we formally introduce the dynamic influence maximization game.  A simplified two-player version of this game was originally proposed in \cite{varma2019allocating},  which can also be viewed as a dynamic version of the game proposed in \cite{christia2021scalable}.

Consider a set $[n]\!=\!\{1,\ldots,n\}$ of individuals that interact over a directed social network $\mathcal{G}=([n],\mathcal{E})$, and a set $[m]=\{1,\ldots,m\}$ of influencers (players). We denote the advertising budget of player $j\in [m]$ by $\beta_j\ge 0$. At any time $t\in [0, t_f]$, where $t_f$ is the final time in the game, we denote the opinion of individual $i\in [n]$ about the players by an $m$-dimensional row vector $\boldsymbol{x}_i(t):=(x_{i1}(t),\ldots,x_{im}(t))\in [0,1]^m$, where $x_{ij}(t)$ is the opinion of individual $i$ about player $j$ at time $t$. We also use the matrix $\boldsymbol{x}(t)\!\in\![0,1]^{n\times m}$ to refer to the opinions of all the individuals at time $t$. 

We assume that there are $K$ advertising campaign times $0\leq t_1<\ldots<t_K\leq t_f$, in which the individuals undergo the influence of the players.  By convention, we define $t_0:=0$ and $t_{K+1}:=t_f$ to refer to initial and final times of the game. Let us denote the remaining advertising budget of player $j$ at campaign time $t_k$ by $\beta_j(t_k)$, where we note that $\beta_j(t_1)=\beta_j$.  At campaign time $t_k, k\in[K]$,  each player $j$ invests according to the (action) vector $\boldsymbol{b}_j(k)=(b_{1j}(k),\ldots,b_{nj}(k))'$,  where $b_{ij}(k)$ is the amount of budget that player $j$ invests on individual $i$ at campaign time $t_k$.  In particular, the action space for player $j$ at time $t_k$ is defined as $\mathcal{B}_j(k):=\{\boldsymbol{b}\in \mathbb{R}_+^n:\sum_{i=1}^n{b}_i\leq \beta_j(k)\}$.  In the absence of the players, i.e., during time intervals $\cup_{k=1}^{K+1}(t_{k-1}, t_k)$,  we assume that individuals' opinions evolve according to the continuous-time DeGroot model \cite{degroot1974reaching},  which is given by the ordinary differential equation $\dot{\boldsymbol{x}}(t)=-L\boldsymbol{x}(t), \boldsymbol{x}(t_0)=\boldsymbol{x}_0$. Here, $\boldsymbol{x}_0$ is the individuals' initial opinions and $L$ is the weighted Laplacian matrix corresponding to the social network $\mathcal{G}=([n],\mathcal{E})$.  As a result of the advertising campaigns at time instances $\{t_k\}_{k=1}^{K}$,  the opinions of the individuals undergo a state jump, which obeys the hybrid model
\begin{align}\nonumber 
&\dot{\boldsymbol{x}}(t)\!=\!-L\boldsymbol{x}(t), \boldsymbol{x}(t_0)\!=\!\boldsymbol{x}_0\ \forall t\in [t_0, t_{K+1}]\!\setminus\! \{t_k\}_{k=1}^{K}\cr 
&x_{ij}(t^+_k)\!=\!\phi\big(x_{ij}(t_k),b_{i1}(k),\ldots,b_{im}(k)\big) \forall i,\!j,\!k\!\in\![K].
\end{align}
Here,  $t^+_k$ is an infinitesimal time after $t_k$,  and the jump function $\phi:\mathbb{R}_+^{m+1}\to [0,1]$ captures the change in individuals' opinions as a result of players' investment strategies at campaign times.  In this work, we consider two specific jump functions that have been well-justified in the past literature both from an axiomatic approach, as well as probabilistic and economic perspectives, \cite{christia2021scalable,varma2018marketing,varma2019allocating,ahmadinejad2015forming}.  The first jump function is for the case where there is only one player,  in which case we can drop the dependency on index $j$ and write
\begin{align}\label{eq:single-player}
\phi\big(x_{i}(t_k),b_{i}(k)\big)=x_i(t_k)+b_i(k),
\end{align}
where $b_i(k)\in [0, 1-x_i(t_k)]$ to assure that the opinions always remain in the interval $[0, 1]$.  Note that \eqref{eq:single-player} simply assumes that the impact of the single influencer's budget allocation on the individuals' opinions is additive.  In a multiplayer
game,  the impact of any player's budget allocation on some individual $i$ will be reduced as other players allocate budget
to $i$.  For that reason,  we assume that the impact of budget allocation on the individuals' opinions
is normalized and given by
\begin{align}\label{eq:multi-player}
\!\!\!\!\phi\big(x_{ij}(t_k\!),b_{i1}(k),\ldots,b_{im}(k)\big)\!=\!\frac{x_{ij}(t_k)\!+\!b_{ij}(k)}{1\!+\!\sum_{\ell=1}^m \!b_{i\ell}(k)}.
\end{align}
Finally,  each player $j$ aims to maximize its average payoff expressed in terms of individuals' opinions:
\begin{align}\label{eq:cumulative-payoff}
U_j(\boldsymbol{b}_j,\boldsymbol{b}_{-j})=\frac{1}{K\!+\!1}\sum_{k=1}^{K+1}\!u_j(\boldsymbol{x}_j(t_k),\boldsymbol{b}_j(k),k),
\end{align} 
where $\boldsymbol{x}_j(t_k)=(x_{1j}(t_k),\ldots,x_{nj}(t_k))'$ is the vector of individuals' opinions about player $j$ at time $t_k$,  $\boldsymbol{b}_j:=(\boldsymbol{b}_j(1),\ldots,\boldsymbol{b}_j(K))$ is the investment strategy of player $j$,\footnote{By convention, we define $\boldsymbol{b}_j(K\!+\!1)\!=\!\boldsymbol{0}$, as it corresponds to the terminal time rather than a campaign time.} and $u_j(\cdot,k):\mathbb{R}^{2n}_+\to \mathbb{R}_+, k\in[K+1]$ are stage utility functions for player $j$.  Therefore, the goal of each player $j$ is to choose an investment strategy $\boldsymbol{b}_j$ to maximize its average payoff \eqref{eq:cumulative-payoff} subject to its budget constraint.  

\begin{definition}
An investment strategy profile $(\boldsymbol{b}_1,\ldots,\boldsymbol{b}_m)$ is called an \emph{open-loop} equilibrium if each player observes the initial opinions and budgets $\boldsymbol{x}(t_0), \boldsymbol{\beta}(t_0)$,  and chooses its investment strategy sequence that satisfies the equilibrium conditions, i.e.,  $U_j(\boldsymbol{b}_j,\boldsymbol{b}_{-j})\ge U_j(\hat{\boldsymbol{b}}_j,\boldsymbol{b}_{-j})$ for any player $j$ and any admissible investment strategy $\hat{\boldsymbol{b}}_j$.
\end{definition}

\section{Single-Player Optimal Strategy with Concave Stage Functions}\label{sec:single-agent}
First, we consider the case of a single-player whose state jump equation is given by \eqref{eq:single-player}.  Since there is only one player, for simplicity of notation, we drop the dependency of the parameters to the player index $j$.  To obtain the optimal investment strategy for the player, one needs to solve the following optimization problem:
\begin{align}\label{eq:single-marketer}
\max\ &\frac{1}{K\!+\!1}\sum_{k=1}^{K+1}u(\boldsymbol{x}(t_k),\boldsymbol{b}(k),k)\cr 
&\dot{\boldsymbol{x}}(t)\!=\!-L\boldsymbol{x}(t), \boldsymbol{x}(t_0)\!=\!\boldsymbol{x}_0 \ t\!\in\! [t_0,\!t_{K\!+\!1}]\!\!\setminus\!\! \{t_k\}_{k=1}^{K}\cr 
&\boldsymbol{x}(t^+_k)=\boldsymbol{x}(t_k)+\boldsymbol{b}(k) \ \forall k\in[K],\cr 
&\boldsymbol{b}(k)\leq \boldsymbol{1}-\boldsymbol{x}(t_k) \ \forall k\in[K],\cr 
&\sum_{k=1}^K\boldsymbol{1}'\boldsymbol{b}(k)\leq \beta, \ \boldsymbol{b}(k)\in \mathbb{R}^n_+\ \forall k\in[K],
\end{align} 
where $\boldsymbol{1}$ is the column vector of all ones,  and the third constraints in \eqref{eq:single-marketer} are componentwise.  Here, $\boldsymbol{x}(t),\boldsymbol{b}(k)\!\in\! \mathbb{R}_+^{n}$, and $u(\cdot,k)\!:\!\mathbb{R}^{2n}_+\!\to \mathbb{R}_+,  k\in [K\!+1]$ are stage utility functions.  As before, $L$ denotes the Laplacian matrix of the underlying social network $\mathcal{G}$, $\beta>0$ is the player's budget,  and $t_{k}, k\in [K]$ are the campaign times that are all known to the player a priori. The goal of the player is to find a sequence of investment vectors $\boldsymbol{b}(k), k\in [K]$ to solve the optimization problem \eqref{eq:single-marketer}. 
\begin{theorem}
Let us assume that the stage utility functions $u(\cdot,k)\!:\!\mathbb{R}^{2n}_+\!\to \mathbb{R}_+,  k\in [K\!+\!1]$ are concave.  Then, an optimal investment strategy for single-player influence maximization \eqref{eq:single-marketer} can be found in polynomial time via a concave program. 
\end{theorem}
\begin{proof}
It is known that the solution to the linear dynamics $\dot{\boldsymbol{x}}(t)=-L\boldsymbol{x}(t),  \boldsymbol{x}(t_0)=\boldsymbol{x}_0$,  is given by $\boldsymbol{x}(t)=e^{-L(t-t_0)}\boldsymbol{x}_0$, where $e^A=\sum_{k=0}^{\infty}\frac{A^k}{k!}$ is the matrix exponential.  By focusing on the $k$th interval $[t_{k-1}, t_k]$ and using the constraints in \eqref{eq:single-marketer}, 
\begin{align}\nonumber
\boldsymbol{x}(t_k)&=e^{-L(t_k-t_{k-1})}\boldsymbol{x}(t^+_{k-1})\cr
&=e^{-L(t_k-t_{k-1})}\big(\boldsymbol{x}(t_{k-1})+\boldsymbol{b}(k\!-\!1)\big)\cr 
&=e^{-L(t_k-t_{k-1})}\boldsymbol{x}(t_{k-1})\!+\!e^{-L(t_k-t_{k-1})}\boldsymbol{b}(k\!-\!1).
\end{align}
Using the above relation recursively, we can express each $\boldsymbol{x}(t_k)$ in terms of budget variables as
\begin{align}\nonumber
\boldsymbol{x}(t_k)&=e^{-L(t_k-t_{k-1})}\boldsymbol{b}(k\!-\!1)\!+\!e^{-L(t_k-t_{k-2})}\boldsymbol{b}(k\!-\!2)\cr 
&\qquad+\ldots+e^{-L(t_k-t_1)}\boldsymbol{b}(1)+e^{-L(t_k-t_0)}\boldsymbol{b}(0),
\end{align}
where by convention we define $\boldsymbol{b}(0)=\boldsymbol{x}_0$. For simplicity of notation, let us define the matrices $A_{rs}=e^{-L(t_r-t_{s})}, \forall 0\leq s<r\leq K+1$,  and note that these matrices are fixed and can be computed by the player a priori.  Then, we have $\boldsymbol{x}(t_k)=\sum_{s=0}^{k-1} A_{ks}\boldsymbol{b}(s), \forall k\in [K+1]$. Using this relation, the optimization problem \eqref{eq:single-marketer} can be written as
\begin{align}\label{eq:concave}
\max\ &\frac{1}{K\!+\!1}\sum_{k=1}^{K+1}u\big(\sum_{s=0}^{k-1} A_{ks}\boldsymbol{b}(s),\boldsymbol{b}(k),k\big)\cr  
&\boldsymbol{b}(k)+\sum_{s=0}^{k-1} A_{ks}\boldsymbol{b}(s)\leq \boldsymbol{1} \ \forall k\in[K],\cr 
&\sum_{k=1}^K\boldsymbol{1}'\boldsymbol{b}(k)\leq \beta,  \ \boldsymbol{b}(k)\!\in\! \mathbb{R}^n_+\ \forall k\in[K].
\end{align} 
Since each stage function $u(\cdot,k), k\in[K+1]$ is assumed to be concave and its arguments are linear with respect to budget variables $\boldsymbol{b}(k), k\in [K]$,  the entire objective function in \eqref{eq:concave} is also concave with respect to budget variables $\boldsymbol{b}=(\boldsymbol{b}(1),\ldots,\boldsymbol{b}(K))$. As a result,  \eqref{eq:concave} is a concave program subject to $2Kn+1$ linear constraints and $Kn$ variables, and hence can be solved in polynomial time.
\end{proof}


\section{Multiplayer Dynamic Influence Maximization Game}\label{sec:multi-player}
In this section, we consider the case where there are multiple competing players whose state jump equation is given by \eqref{eq:multi-player}.  More precisely,  as a result of advertising campaigns, the individuals' opinions now follow the hybrid model:
\begin{align}\label{eq:multi-marketer}
&\!\!\!\!\dot{\boldsymbol{x}}(t)\!=\!-L\boldsymbol{x}(t), \boldsymbol{x}(t_0)\!=\!\boldsymbol{x}_0,  t\!\in\! [t_0,\!t_{K\!+\!1}]\!\!\setminus\!\! \{t_k\}_{k=1}^{K}\cr 
&\!\!\!\!x_{ij}(t^+_k)=\frac{x_{ij}(t_k)+b_{ij}(k)}{1+\sum_{\ell=1}^m b_{i\ell}(k)}, \ \forall i,j,k\in[K],
\end{align} 
where we recall that $b_{ij}(k)$ denotes the advertising budget of player $j$ on individual $i$ at campaign time $t_k$.  Note that since there are $m$ different players, the opinion of individual $i$ is an $m$-dimensional row vector $\boldsymbol{x}_i(t)\in [0,1]^m$. In particular,  $\boldsymbol{x}(t)$ is an $n\times m$ matrix whose $i$th row equals $\boldsymbol{x}_i(t)$.  Also, notice that if the opinion vector of individual $i$ lies in the probability simplex before the players allocate budget to that individual, then the update rule \eqref{eq:multi-marketer} guarantees that the opinion vector will still lie in the probability simplex post budget allocation.

For any $k\in[K]$, let us define the diagonal matrix $D(k)=\mbox{diag}(\frac{1}{1+\sum_{\ell=1}^m b_{1\ell}(k)},\ldots,\frac{1}{1+\sum_{\ell=1}^m b_{n\ell}(k)})$, and the $n\times m$ budget matrix $B(k)=(b_{ij}(k))_{i,j}$.  Note that the $j$th column of $B(k)$ is precisely the investment strategy of player $j$ at time $t_k$, that is $B(k)=(\boldsymbol{b}_1(k)|\ldots|\boldsymbol{b}_m(k))$.  We can now write the state jump equation in \eqref{eq:multi-marketer} in a compact form as
\begin{align}\nonumber
\boldsymbol{x}(t_k^+)=D(k)\boldsymbol{x}(t_k)+D(k)B(k),  \ \forall k\in [K].
\end{align} 
Using an inductive argument as in the case of a single player, and because the multiplayer opinion dynamics can be decomposed into $m$ separate single-player dynamics, i.e., $\dot{\boldsymbol{x}}_j(t)=-L\boldsymbol{x}_{j}(t)\ \forall j\in [m]$,  for any $k\in [K+1]$ we can write 
\begin{align}\nonumber
&\boldsymbol{x}(t_k)=e^{-L(t_k-t_{k-1})}\boldsymbol{x}(t^+_{k-1})\cr 
&=A_{k,k-1}(D(k-1)\boldsymbol{x}(t_{k-1})+D(k-1)B(k-1))\cr 
&=A_{k,k-1}D(k-1)A_{k-1,k-2}D(k-2)\boldsymbol{x}(t_{k-2})\cr 
&\qquad+A_{k,k-1}D(k-1)A_{k-1,k-2}D(k-2)B(k-2)\cr
&\qquad+A_{k,k-1}D(k-1)B(k-1)=\cdots\cr
&=\sum_{s=0}^{k-1}\Big(\prod_{r=s}^{k-1}A_{r+1,r}D(r)\Big)B(s),
\end{align}
where by convention $D(0)=I_{n\times n}, B(0)=\boldsymbol{x}_0$, and the product in the above expression multiplies the matrices from the left side. Thus, given fixed strategy of all others, the optimization problem that the $j$-th player faces can be formulated as:
\begin{align}\label{eq:x-expression}
\max\ &\frac{1}{K\!+\!1}\sum_{k=1}^{K+1}u_j(\boldsymbol{x}_j(t_k),\boldsymbol{b}_j(k),k)\cr
&\boldsymbol{x}_j(t_k)=\sum_{s=0}^{k-1}\Big(\prod_{r=s}^{k-1}A_{r+1,r}D(r)\Big)\boldsymbol{b}_j(s)\ \forall k,\cr 
&\sum_{k=1}^{K}\boldsymbol{1}'\boldsymbol{b}_j(k)\leq \beta_{j}, \ \boldsymbol{b}_j(k)\!\in\! \mathbb{R}^n_+\ \forall k\!\in\![K].
\end{align}

Next, we consider the following definition from online convex optimization.
\begin{definition}
Given a convex set $\mathcal{C}$ and a sequence of concave functions $f^\tau:\mathcal{C}\to \mathbb{R}, \tau\in [T]$,  let $\mathcal{A}$ be an online algorithm that at each time $\tau$ selects a point $x^{\tau}\in \mathcal{C}$, and let $x$ be the optimal static solution, i.e.,  $x=\argmax_{y\in \mathcal{C}} \sum_{\tau=1}^{T}f^{\tau}(y)$. Then, the regret of the algorithm $\mathcal{A}$ is defined by $\mathcal{R}_{\mathcal{A}}(T)=\sum_{\tau=1}^{T}f^{\tau}(x)-\sum_{\tau=1}^{T}f^{\tau}(x^{\tau})$.  An algorithm $\mathcal{A}$ has no regret, if $\mathcal{R}_{\mathcal{A}}(T)=o(T)$.
\end{definition}

\begin{theorem}\label{thm:multi-player}
Let $u_{j}(\cdot,\cdot,k)\!:\![0,1]^n\!\times\! \mathbb{R}_+^{n} \!\to\! \mathbb{R}$, $\forall k,j$, be twice differentiable functions that are increasing and convex with respect to their first argument.  If there are positive constants $\lambda_j$ such that $\sum_{j=1}^m\lambda_jU_j(\boldsymbol{b}_j,\boldsymbol{b}_{-j})$ is a concave function in $\boldsymbol{b}=(\boldsymbol{b}_j,\boldsymbol{b}_{-j})$, then the game admits a pure-strategy open-loop equilibrium. Moreover, any no-regret algorithm that players use in the repeated version of the game will converge to one such equilibrium.
\end{theorem}
\begin{proof}
Using Lemma \ref{lemm-convexity} in Appendix I,  the total utility of player $j$ that is given by $U_j(\boldsymbol{b}_j,\boldsymbol{b}_{-j})=\frac{1}{K+1}\sum_{k=1}^{K+1}u_j(\boldsymbol{x}_j(t_k),\boldsymbol{b}_j(k),k)$,  is a convex function with respect to other players' investement strategies $\boldsymbol{b}_{-j}$. Moreover,  by the assumption, the social utility $\sum_{j=1}^m\lambda_jU_j(\boldsymbol{b}_j,\boldsymbol{b}_{-j})$ is a concave function in $\boldsymbol{b}$. Therefore, the dynamic influence maximization game is a socially concave game \cite[Definition 2.1]{even2009convergence}. Since each player's utility function is twice differentiable,  the game is also a concave game and hence admits a pure-strategy Nash equilibrium \cite[Lemma 2.2]{even2009convergence}.  Finally,  using \cite[Theorem 3.1]{even2009convergence},  if every player $j$ plays according to a no-regret algorithm with a regret of $\mathcal{R}(T)$, then the average strategy vector $\frac{1}{T}\sum_{\tau=1}^T\boldsymbol{b}^{\tau}$ will converge to a pure-strategy Nash equilibrium at a rate of $O(\frac{m\mathcal{R}(T)}{T})$. Thus, to obtain open-loop equilibrium strategies, each player $j$ needs to play a no-regret algorithm with constraint set $\mathcal{C}_j=\{\boldsymbol{b}_j\in \mathbb{R}_+^{nK}: \sum_{k=1}^K\boldsymbol{1}'\boldsymbol{b}_j(k)\leq \beta_j\}$, and the sequence of concave functions $f^{\tau}(\boldsymbol{b}_j):=U_j(\boldsymbol{b}_j,\boldsymbol{b}^{\tau}_{-j})$, where $\boldsymbol{b}^{\tau}_{-j}$ is the joint strategy played by other players at iteration $\tau$.  
\end{proof}

It is often reasonable to assume that the influence maximization game is a constant-sum (hence a socially concave) game, i.e., $\sum_{j=1}^m U_j(\boldsymbol{b})=c$, as players are competing for \emph{fixed} units of individuals' opinions.  More precisely, according to  \eqref{eq:multi-marketer}, an individual distributes its unit opinion fractionally among the players at each campaign time.  Thus,  assuming some linearity of the stage functions,  the total utility that all players can derive  is a constant function of $n, T$, and $\beta_j$ (see Section \ref{sec:numerical}). The result can also be extended to socially concave games as long as the players' stage utilities are linear in opinions $\boldsymbol{x}(t_k)$ and concave with respect to investment strategies $\boldsymbol{b}(k)$.

\begin{remark}\label{rem:final}
There are many no-regret algorithms in the past literature.  For instance,  for strictly concave stage utilities, the online gradient ascent algorithm in \cite{hazan2007logarithmic} achieves a regret bound of $\mathcal{R}(T)=O(\log T)$.  Thus, if the players use such an algorithm, their average strategies converge to an open-loop equilibrium at the speed of $O(\frac{m\log T}{T})$. 
\end{remark}

\subsection{A Numerical Example}\label{sec:numerical}
Let us consider the special two-player game defined in \cite{varma2019allocating},  where the stage utilities are given by linear functions $u_1(\boldsymbol{x}(t_k),\boldsymbol{b}_1(k),k)=\rho(k)'\boldsymbol{x}(t_k)-\lambda_1 \boldsymbol{1}'\boldsymbol{b}_1(k)$ and $u_2(\boldsymbol{1}-\boldsymbol{x}(t_k),\boldsymbol{b}_2(k),k)=\rho(k)'(\boldsymbol{1}-\boldsymbol{x}(t_k))-\lambda_2 \boldsymbol{1}'\boldsymbol{b}_2(k)$. Here, $\rho(k)$ is a nonnegative constant vector and $\lambda_1,\lambda_2$ are two positive constants.  Clearly, these functions are twice differentiable,  and strictly convex and increasing with respect to their first argument.  Moreover,  $U_1(\boldsymbol{b}_1,\!\boldsymbol{b}_{2})+\!U_2(\boldsymbol{b}_1,\boldsymbol{b}_{2})\!=\!\frac{1}{K\!+\!1}(\sum_{k=1}^{K\!+\!1}\rho(k)'\boldsymbol{1}\!-\!\lambda_1\beta_1\!-\!\lambda_2\beta_2)$, which is a constant (and hence a concave) function. Therefore, all the conditions of Theorem \ref{thm:multi-player} are satisfied,  which in view of Remark \ref{rem:final} implies that if both players follow a no-regret gradient ascent algorithm, their average strategies will converge to an open-loop equilibrium at a rate of $O(\frac{\log T}{T})$.

\begin{figure}[t]
\begin{center}
\vspace{-3cm}
\hspace{-2.5cm}
\includegraphics[totalheight=.3\textheight,
width=.4\textwidth,viewport=0 0 400 400]{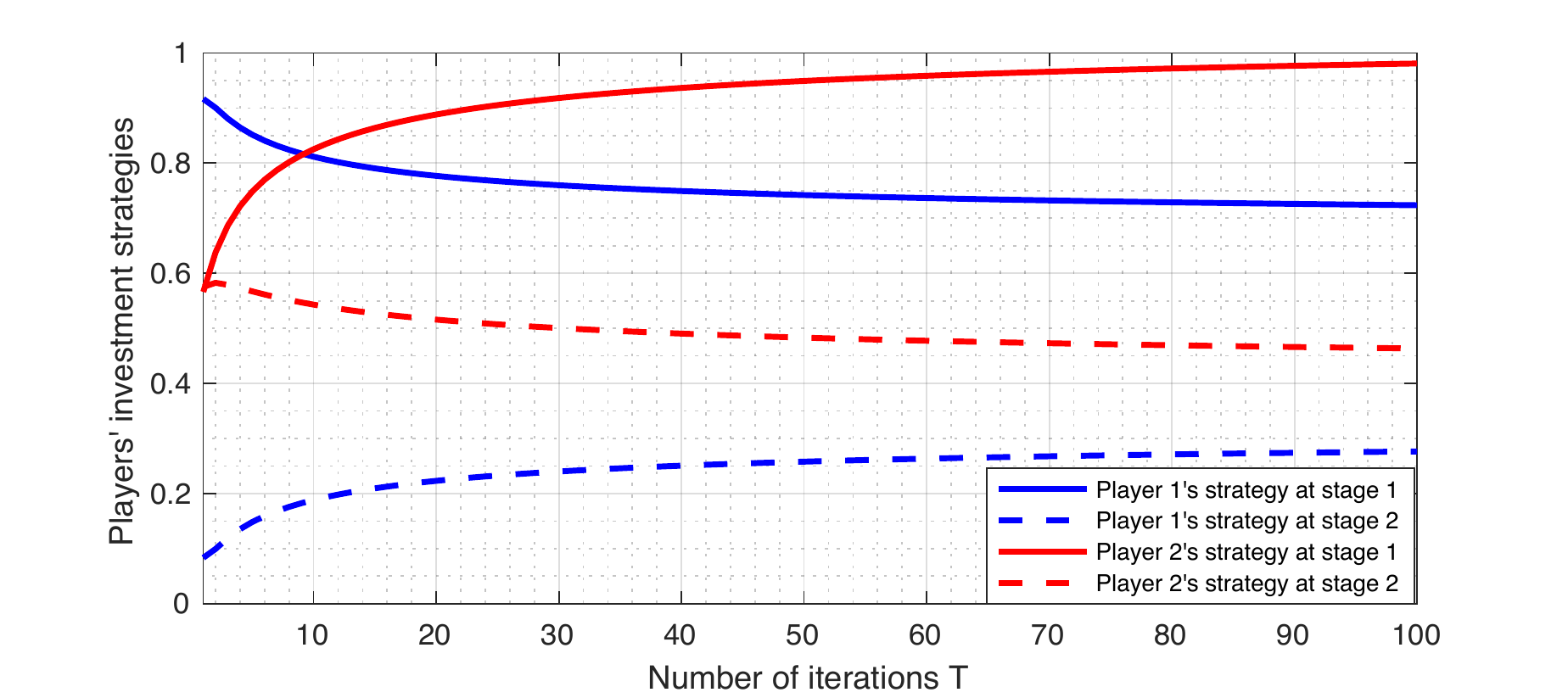} \hspace{0.4in}
\end{center}
\vspace{-0.5cm}
\caption{Players' investment strategies asymptotically converge to their equilibrium values.  At each stage, players invest the same amount in all the individuals, whose values are depicted by solid lines (first stage) and dashed lines (second stage). Blue/red curves correspond to the strategy of the first/second player, respectively. }\label{fig:NE}\vspace{-0.5cm}
\end{figure}

To justify the above analysis numerically, we consider $n=3$ individuals that are connected using a path with Laplacian $L={\tiny\begin{bmatrix}1/3 &-1/3& 0\\ -1/3 & 2/3 & -1/3\\ 0 & -1/3 & 1/3
\end{bmatrix}}$. We let $K=2$, $\lambda_1=\lambda_2=1$,  $\rho(1)=\rho(2)=\boldsymbol{1}$, $\beta_1=3, \beta_2=5$. Finally, we choose the campaign times to be $t_1=1, t_2=2$, and the initial and terminal times to be $t_0=0, t_3=3$. Using these parameters, the players' utilities simplify to $U_1(\boldsymbol{b}_1,\boldsymbol{b}_{2})=\frac{1}{3}\sum_{k=1}^3\boldsymbol{1}'(\boldsymbol{x}(k)-\boldsymbol{b}_{1}(k))$, and $U_2(\boldsymbol{b}_1,\boldsymbol{b}_{2})=3-\frac{1}{3}\sum_{k=1}^3\boldsymbol{1}'(\boldsymbol{x}(k)+\boldsymbol{b}_{2}(k))$,
subject to the state jump dynamics \eqref{eq:multi-marketer} and initial opinions $\boldsymbol{x}_0:=\frac{1}{2}\boldsymbol{1}$.  We run the no-regret gradient ascent dynamics $\boldsymbol{b}_j^{\tau+1}=\Pi_{\mathcal{C}_j}[\boldsymbol{b}^{\tau}_j+\eta_{\tau} \nabla_{\boldsymbol{b}_j} U_j(\boldsymbol{b}^{\tau}_j,\boldsymbol{b}^{\tau}_{-j})]$ for both players $j=1,2$, with stepsize $\eta_{\tau}=\frac{10}{\tau}$,  and for $T=100$ iterations.\footnote{Here, $\Pi_{\mathcal{C}_j}[\cdot]$ denotes the $\ell_2$-norm projection operator on the budget constraint set $\mathcal{C}_j=\{\boldsymbol{b}_j\in \mathbb{R}_+^{6}: \sum_{k=1}^2\boldsymbol{1}'\boldsymbol{b}_j(k)\leq \beta_j\}$. } The result of average investment strategies $\frac{1}{T}\sum_{\tau=1}^T\boldsymbol{b}_j^{\tau}$ for player 1 and player 2 are illustrated in Figure \ref{fig:NE} by blue and red curves, respectively.  In this figure, the solid/dashed lines correspond to the investment strategies at the first/second campaign time, where the amount of investment on all the individuals at each stage is the same.  As can be seen, the investment curves converge asymptotically at a rate of $O(\frac{\log T}{T})$ to their equilibrium values.  Therefore,  players can run these simulations at the beginning of the game, and the limit values will give them their open-loop equilibrium strategies.

\section{Conclusions}\label{sec:conclusion}

We studied a dynamic influence maximization game that provides a natural paradigm to model competitive advertising over social networks. We considered this problem from both optimization and game-theoretic perspectives and developed efficient algorithms for computing optimal investment strategies and open-loop Nash equilibria.  One interesting future direction is to devise an efficient algorithm for computing feedback Nash equilibria.


\bibliographystyle{IEEEtran}
\bibliography{thesisrefs}

\section*{Appendix I}

In this Appendix, we complete the proof of Theorem by showing that the utility of each player is a convex function of other players' investment strategies (Lemma \ref{lemm-convexity}).  But before that, we first show the following two auxiliary lemmas. 

\begin{lemma}\label{lemma:basic}
Let $\boldsymbol{w}_r\in \mathbb{R}_+^m, r\in [d]$ be nonnegative weight vectors for some $d\in\mathbb{Z}_+$, and $a_r> 0, r\in [d]$ be positive constants. Then, the multivariable function $h(\boldsymbol{y}):=\prod_{r=1}^d\big(\frac{1}{a_r+\boldsymbol{w}'_r\boldsymbol{y}_r}\big)$ is a (jointly) convex function of $\boldsymbol{y}=(\boldsymbol{y}_1,\ldots,\boldsymbol{y}_d)\in \mathbb{R}^{dm}_+$.  
\end{lemma}
\begin{proof}
Let us define $f:\mathbb{R}_+^d\to \mathbb{R}_+$ by $f(\boldsymbol{x})=\prod_{r=1}^d(\frac{1}{a_r+x_r})$, and consider the function $g(\boldsymbol{x}):=\ln f(\boldsymbol{x})=\sum_{r=1}^d-\ln(a_r+x_r)$. Clearly, each summand $-\ln(a_r+x_r)$ is a convex function as $\frac{d^2}{dx_r^2}(-\ln(a_r+x_r))=\frac{1}{(a_r+x_r)^2}>0$. Therefore, $g(\boldsymbol{x})$ is a convex function. Since the composition of an increasing convex function and a convex function is also convex, $f(\boldsymbol{x})=e^{g(\boldsymbol{x})}$ is a convex function.  Finally, we note that since $h(\boldsymbol{y})=f(\boldsymbol{w}_1'\boldsymbol{y}_1,\ldots,\boldsymbol{w}'_d\boldsymbol{y}_d)$, the function $h$ is convex with respect to $\boldsymbol{y}$, as it is obtained from composition of a convex function and several linear functions. More precisely, for any $\lambda\in [0,1]$, we have

\vspace{-0.5cm}
{\small \begin{align}\nonumber
&h(\lambda \boldsymbol{y}+(1-\lambda)\hat{\boldsymbol{y}})\cr 
&=f\big(\lambda (\boldsymbol{w}_1'\boldsymbol{y}_1,\ldots,\boldsymbol{w}'_d\boldsymbol{y}_d)+(1-\lambda)(\boldsymbol{w}_1'\hat{\boldsymbol{y}}_1,\ldots,\boldsymbol{w}'_d\hat{\boldsymbol{y}}_d)\big)\cr 
&\leq \lambda f\big(\boldsymbol{w}_1'\boldsymbol{y}_1,\ldots,\boldsymbol{w}'_d\boldsymbol{y}_d\big)+(1-\lambda) f\big(\boldsymbol{w}_1'\hat{\boldsymbol{y}}_1,\ldots,\boldsymbol{w}'_d\hat{\boldsymbol{y}}_d\big)\cr
&=\lambda h(\boldsymbol{y})+(1-\lambda)h(\hat{\boldsymbol{y}}), 
\end{align}}
where the inequality uses the convexity of $f$.
\end{proof}

\begin{lemma}\label{lemm:stochastic}
Let $L=I-A_{\mathcal{G}}$ be the Laplacian of the network $\mathcal{G}=([n],\mathcal{E})$, where $A_{\mathcal{G}}=(a_{ij})_{i,j}$ is the stochastic weighted adjacency matrix of $\mathcal{G}$. Then, $e^{-Lt}$ is a stochastic matrix for any $t\ge 0$.
\end{lemma}
\begin{proof}
Notice that all the rows of $e^{-Lt}$ sum to $1$ because $e^{-Lt}\boldsymbol{1}=I\boldsymbol{1}+\sum_{k=1}^{\infty}\frac{t^k}{k!}(I-A_{\mathcal{G}})^{k}\boldsymbol{1}=\boldsymbol{1}$, where the second equality holds because $(I-A_{\mathcal{G}})\boldsymbol{1}=0$ by stochasticity of $A_{\mathcal{G}}$.  Moreover,  for any $\boldsymbol{x}_0\in[0,1]^n$, the solution to the dynamics $\dot{\boldsymbol{x}}(t)\!=\!-L\boldsymbol{x}(t),  \boldsymbol{x}(0)=\boldsymbol{x}_0$ is nonnegative.  The reason is that the solution trajectories are continuous and given by $\boldsymbol{x}(t)=e^{-Lt}\boldsymbol{x}_0$. So if some trajectory $x_i(t)$ becomes negative for the first time, there must be a time $\hat{t}\ge 0$ for which $x_i(\hat{t})=0$. However at that time we have $\dot{x}_i(\hat{t})=-(1-a_{ii})x_i(\hat{t})+\sum_{\ell\neq i}a_{i\ell}x_{\ell}(\hat{t})=\sum_{\ell\neq i}a_{i\ell}x_{\ell}(\hat{t})\ge 0$.  Since $\dot{x}_i(\hat{t})\ge 0$, the trajectory $x_i(t)$ is nondecreasing at time $\hat{t}$, and hence cannot go negative.  Thus, for any $\boldsymbol{x}_0\in[0,1]^n$ and $t\ge 0$, all the coordinates of $e^{-Lt}\boldsymbol{x}_0$ are nonnegative, and hence the matrix $e^{-Lt}$ has only nonnegative entries.   
\end{proof}

\begin{lemma}\label{lemm-convexity}
Let $u_{j}(\cdot,\cdot,k):[0,1]^n\times \mathbb{R}_+^{n} \to \mathbb{R}, \forall k$ be twice differentiable functions that are increasing and convex with respect to their first argument. Then for any $k\in [K+1]$,  the function $u_j(\boldsymbol{x}_j(t_k),\boldsymbol{b}_j(k),k)$ is twice differentiable, and it is convex with respect to other players' strategies.
\end{lemma}
\begin{proof}
If we fix player $j$-th strategy to $\boldsymbol{b}_{j}=(\boldsymbol{b}_{j}(k))_{k\in [K]}$, then $\boldsymbol{x}_j(t_k)$ is a vector function of other players' strategies $\boldsymbol{b}_{-j}=(\boldsymbol{b}_{\ell}(k))_{\ell\neq j, k\in [K]}$, i.e., $\boldsymbol{x}_j(t_k)\!:\!\boldsymbol{b}_{-j}\!\to\! [0,\!1]^n\!$.  Therefore, if we can show that each coordinate of $\boldsymbol{x}_j(t_k)$ is a convex function of $\boldsymbol{b}_{-j}$, then the convexity of $u_j(\boldsymbol{x}_j(t_k),\boldsymbol{b}_j(k),k)$ as a function of $\boldsymbol{b}_{-j}$ follows using the composition of the increasing convex function $u_j(\cdot,\boldsymbol{b}_j(k),k)$ and the convex vector function $\boldsymbol{x}_j(t_k)$.   

To show that each coordinate of $\boldsymbol{x}_{j}(t_k)$ is a convex function of $\boldsymbol{b}_{-j}$, note that the vectors $\boldsymbol{b}_j(s), s\in [K]$ are determined by the strategy of player $j$, and hence can be viewed as constants in the convexity analysis with respect to $\boldsymbol{b}_{-j}$.  Let us fix an arbitrary $s\in [K]$, and consider the summand $\big(\prod_{r=s}^{k-1}A_{r+1,r}D(r)\big)\boldsymbol{b}_j(s)$ in the definition of $\boldsymbol{x}_j(t_k)$ given in \eqref{eq:x-expression}. Since by Lemma \ref{lemm:stochastic} the matrices $A_{r+1,r}$ contain only nonnegative entries and $\boldsymbol{b}_j(s)$ are nonnegative constant vectors, each coordinate of $\big(\prod_{r=s}^{k-1}A_{r+1,r}D(r)\big)\boldsymbol{b}_j(s)$ can be expressed as
\begin{align}\label{eq:expanded-matrix}
\sum_{(i_s,\ldots,i_{k-1})\in [n]^{k-s}}p_{i_s,\ldots,i_{k-1}}\prod_{r=s}^{k-1}\frac{1}{1+\boldsymbol{1}'\boldsymbol{b}_{i_r}(r)},
\end{align} 
where $p_{i_s,\ldots,i_{k-1}}\ge 0$ are some nonnegative (possibly zero) coefficients and by some abuse of notation $\boldsymbol{b}_{i_r}(r):=(b_{i_r1}(r),\ldots,b_{i_rm}(r))'$. The reason is that each matrix $A_{r+1,r}D(r)$ in the product $\prod_{r=s}^{k-1}A_{r+1,r}D(r)$ contributes exactly one term of the form $\frac{1}{1+\boldsymbol{1}'\boldsymbol{b}_{i_r}(r)}$ for some $i_r\in [n]$ to the final expression obtained after matrix multiplication. Since the multiplying variables $\boldsymbol{b}_{i_r}(r)$ in the expression \eqref{eq:expanded-matrix} do not share the same time index $r$, they are independent and we can use Lemma \ref{lemma:basic} with $\boldsymbol{w}_r=\boldsymbol{1} \ \forall r$, to conclude that each of the summands in \eqref{eq:expanded-matrix}, and hence the entire expression \eqref{eq:expanded-matrix} is a convex function of $\boldsymbol{b}_{-j}$. Repeating the same argument for each $s=0,\ldots,k-1$, and summing all the terms shows that each coordinate of $\boldsymbol{x}_j(t_k)$ is a convex function of $\boldsymbol{b}_{-j}$.  Finally,  since each summand in \eqref{eq:expanded-matrix} is twice diffentiable over the positive orthant,  each coordinate of $\boldsymbol{x}_j(t_k)$ and hence $u_j(\boldsymbol{x}_j(t_k),\boldsymbol{b}_j(k),k)$ is a twice differentiable function.
\end{proof}

\end{document}